\documentclass[journa,comsocl]{IEEEtran}
\include{graphics}
\include{amsmath}
\usepackage{color,soul}
\setstcolor{red}
\usepackage{multirow}
\usepackage{algorithm}
\usepackage{algpseudocode}
\usepackage{verbatim}
\usepackage[latin1]{inputenc}
\usepackage{amsmath}
\usepackage{amsfonts}
\usepackage{amssymb}
\usepackage{mathrsfs}
\usepackage{booktabs}
\usepackage{url}
\usepackage[skip=1pt,font=scriptsize]{subcaption}
\usepackage{ifthen}
\usepackage{xspace}
\usepackage{dsfont}
\usepackage{bbm}
\usepackage{cite}
\usepackage{epsfig}
\usepackage{array}
\usepackage{multicol}
\usepackage{algpseudocode}
\usepackage{algorithm}
\usepackage{amssymb}
\usepackage{breqn}
\usepackage{esint}
\usepackage[skip=2pt,font=footnotesize]{caption}
\usepackage{multicol}
\usepackage{amsthm}
\usepackage{setspace}
\usepackage{lipsum}
\usepackage{dblfloatfix}

\theoremstyle{plain}
\newtheorem{theorem}{Theorem}
\newtheorem{lemma}{Lemma}
\theoremstyle{definition}

\theoremstyle{remark}

\usepackage{footnote}
\newcommand\numberthis{\addtocounter{equation}{1}\tag{\theequation}}

\renewcommand{\baselinestretch}{0.98}

\begin{document}
%
\title{Throughput Optimal Beam Alignment in \\ Millimeter Wave Networks}
\author{Muddassar Hussain and Nicolo Michelusi
\thanks{M. Hussain and N. Michelusi are with the School of Electrical and Computer Engineering, Purdue University. \emph{email}: \{hussai13,michelus\}@purdue.edu.}
\thanks{This research has been funded by NSF under grant CNS-1642982.}
\vspace{-7mm}
}
\maketitle
\begin{abstract}
Millimeter wave communications rely on narrow-beam transmissions to cope with the strong signal attenuation at these frequencies, thus demanding precise beam alignment between transmitter and receiver.
The communication overhead incurred to achieve beam alignment may become a severe impairment in mobile networks. This paper addresses the problem of optimizing beam alignment acquisition, with the goal of maximizing throughput. Specifically, the algorithm jointly determines the portion of time devoted to beam alignment acquisition, as well as, within this portion of time, the optimal beam search parameters, using the framework of Markov decision processes. It is proved that a \emph{bisection search} algorithm is optimal, and that it outperforms
exhaustive and iterative search algorithms proposed in the literature.
The duration of the beam alignment phase is optimized so as to maximize the overall throughput.
The numerical results show that the
throughput, optimized with respect to the duration of the beam alignment phase, achievable under the exhaustive algorithm is 88.3\% lower than that achievable under the bisection algorithm. Similarly, the throughput achievable by the iterative search algorithm for a division factor of 4 and 8 is, respectively, 12.8\% and 36.4\% lower than that achievable by the bisection algorithm.  
\end{abstract}
\vspace{-3mm}
\section{introduction}
Mobile data traffic has shown a tremendous growth in the past few decades. Over the last decade alone,  mobile data traffic has grown 4000-fold and is expected to increase by $53\%$ in each year until 2021 \cite{cisco2}. Traditionally, mobile data traffic is served almost exclusively by wireless systems operating under 6 GHz, due to the availability of low-cost hardware and favorable propagation characteristics at these frequencies. However, many current and future applications, such as virtual/augmented reality, high definition video streaming, will require a much higher data rate, which cannot be supported by sub-6 GHz networks due to limited bandwidth availability. 

\par Recently, there has been increasing interest in the research community in developing systems utilizing frequencies in the 28-100 GHz range, the so called 
millimeter wave (mm-wave) frequencies, as a way to alleviate the spectrum crunch \cite{rappaport,gosh,MicheICC}.
This increased interest can be attributed to the availability of larger bandwidth in the mm-wave frequency band, which can better address the demands of the ever increasing mobile traffic. According to Friis' law, at the mm-wave frequency a higher isotropic path loss is incurred compared to sub-6 GHz systems \cite{rappaport_book}. In order to overcome these challenging channel conditions, mm-wave communications are expected to leverage narrow-beam communications \cite{channel_model}, hence both base stations and mobile devices will be equipped with many antennas with multiple-input multiple-output processing, such as precoding, beamforming, and combining to achieve directionality and alleviate the propagation loss at these frequencies.
\par However, maintaining beam alignment between transmitter and receiver in mm-wave networks can be very challenging, especially in mobile scenarios. The resulting communication overhead may thus become the bottleneck of the system. Hence, it is imperative to optimize the beam alignment algorithm to minimize the communication overhead, while optimizing network performance such as delay, or throughput.
Motivated by this challenge, this paper addresses optimal design of the beam alignment algorithm to maximize throughput.

\par  In the literature, the issue of beam alignment has been partly studied under the topic of initial access, \emph{i.e.}, the procedure by which a mobile user (MU) discovers and connects to a mm-wave base station (BS) [\citen{exhaustive}--\citen{jandrew}]. While the initial access is a simple task in legacy cellular systems such as LTE, it becomes a a challenging task in mm-wave networks since not only the MU has to discover the base station using directional beams, but also the MU and BS need to agree upon a beam pattern to be used for future communications. To this end, several schemes for the initial access in mm-wave networks have been proposed in [\citen{exhaustive}--\citen{jandrew}]. One of the most popular ones is called the \emph{exhaustive search}, whereby the BS and the MU sequentially search through all possible combinations of transmit and receive beam patterns \cite{exhaustive}. An \emph{iterative search} algorithm is proposed in \cite{iterative}, where the BS first searches in wider sectors by using wider beams, and then refines the search within the best sector. Similarly, in \cite{two-step}, a two-step initial search procedure is proposed, where the macro BS disseminates the GPS coordinates of the BSs in the vicinity omni-directionally to the MUs and an MU decides a beamforming pattern for the best BS by using its own GPS coordinates followed by an exhaustive search by the BS. In [\citen{zorzi},\citen{jandrew}], different variants of exhaustive search are studied. Specifically, link level performance of different variants of exhaustive search is studied in \cite{zorzi}, while \cite{jandrew} studies network wide performance of these variants using stochastic geometry. It should be noted that these variants of exhaustive search algorithms result from different combinations of directional and omni-directional beamforming at the BS and MU.
\begin{figure}[!t]
\centering
\vspace{-3mm}
\includegraphics[width=0.47\columnwidth]{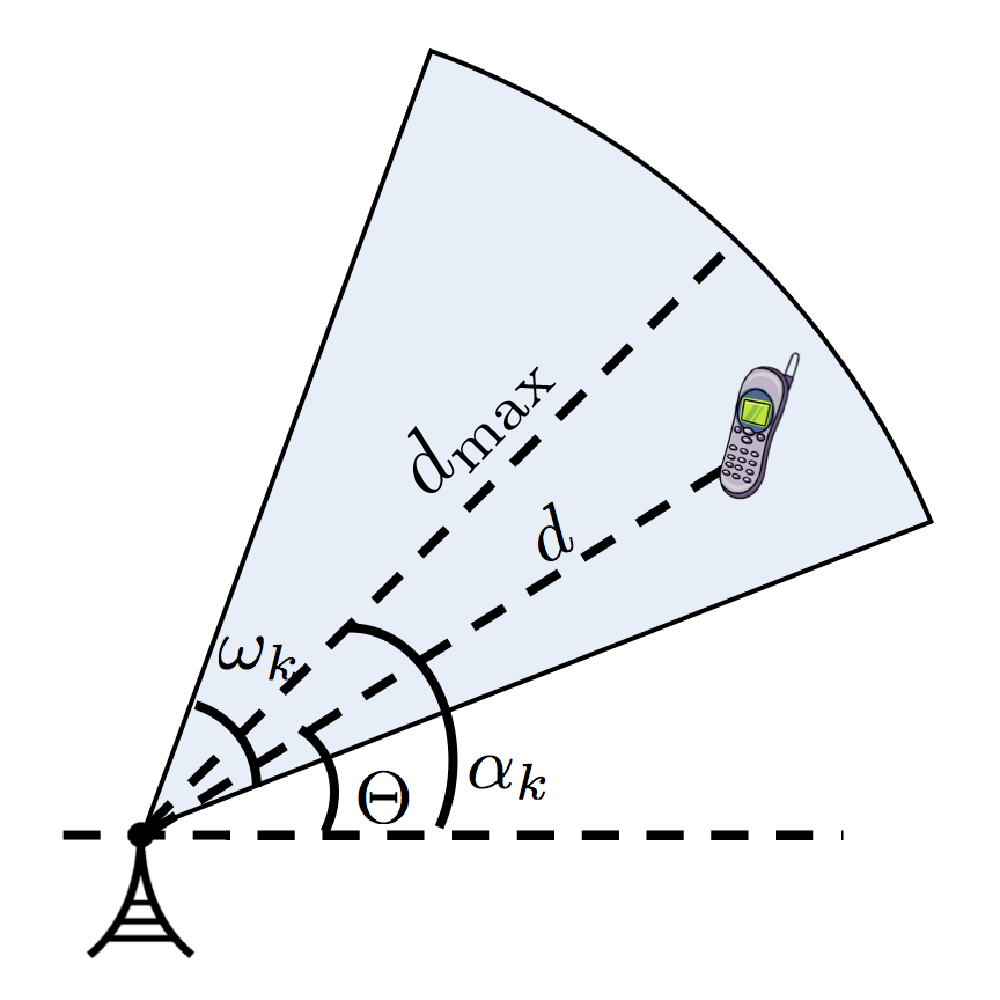}
\caption{The beam pattern
under the sectored antenna model \cite{sectored_model}.}
\label{fig:beam_pattern}
\vspace{-6mm}
\end{figure}
\par In all of the aforementioned papers, the optimality of the search algorithms is not established.  In this paper, we design a beam alignment protocol with the goal of maximizing the throughput to the MU. Specifically, we consider a time-slotted system and focus on downlink. We allocate a fraction of the frame length to the sensing/search phase and the remaining slots for communication. We assume that the MU receives omni-directionally and the BS transmits a number of sensing beacons with varying directional beam patterns in the sensing phase to detect the MU with the goal to maximize the throughput in the data communications period. We use a Markov decision process (MDP) formulation to model the sensing phase and find an optimal sensing policy in closed-form which maximizes the downlink throughput.
We prove that the iterative and exhaustive algorithms proposed in the literature are suboptimal and that, instead, 
a bisection search algorithm is optimal. Moreover, 
we optimize the duration of the sensing phase to maximize the overall throughput.
We show numerically that the
throughput, optimized with respect to 
the duration of the sensing phase,
achievable under the bisection algorithm outperforms by 88.3\% that achievable by the 
exhaustive search algorithm, and by 12.8\% and 36.4\% that achievable under the iterative search algorithm, with division factor of 4 and 8, respectively.
\vspace{-3mm}
\section{System Model}
We consider a mm-wave based cellular system with one base station (BS) and one mobile user (MU), as shown in Fig.~\ref{fig:beam_pattern}. 
Time is slotted with slot duration $T$ [seconds].
It is assumed that the BS is located at the origin and the MU is located at polar coordinates $(d,\Theta)$ with respect to the BS,
where $d\in (0,d_{\max})$ is the distance from the BS,
$d_{\max}$ is the coverage area of the BS,
and $\Theta\in [-\pi,\pi)$ is the angular coordinate,
as shown in Fig.~\ref{fig:beam_pattern}; we assume that $\Theta$ is uniformly distributed in $[-\sigma/2,\sigma/2]$, \emph{i.e.}, $\Theta\sim \textnormal{Uniform}[-\sigma/2,\sigma/2]$,
where $\sigma\in(0,2\pi]$.
In this paper, we approximate the transmission beam of the BS using the \emph{sectored antenna model} \cite{sectored_model}, as
 depicted in Fig.~\ref{fig:beam_pattern}. Thus,  $\omega_k$ and $\alpha_k$ denote the beam-width and angle of departure in slot $k$, respectively. It should be noted that we ignore the effect of secondary beam lobes. Moreover, it is assumed that the MU receives isotropically.
 
 We now introduce the beam alignment protocol, whose aim is to optimize the alignment between transmitter and receiver by leveraging the directionality of mm-wave transmissions.
Beam alignment, herein also termed "sensing",
 and data communication are performed in an alternating fashion. An abstract timing diagram of the MU  beam alignment protocol is shown in Fig.~\ref{fig:timing_diagram}, which illustrates both the sensing and data communication phases. We assume that one frame has duration $N=L+M\geq 1$ and comprehends an initial sensing phase, of duration $L$ slots, with $L\in\{0,1,\dots, N\}$,
 followed by a data communication phase, of duration $M=N-L$ slots.
In the beginning of each time slot $k$ during the sensing period, the BS sends a beacon $b_k$ with beam parameters $(P_{TX,k},\alpha_{k},\omega_{k})$
and of duration $T_B<T$
to detect the MU, where $P_{TX,k}$ denotes the transmission power of the BS.

In this paper, we assume that $P_{TX,k}$ is chosen such that the  signal-to-noise ratio (SNR) measured at the MU at any distance $d\leq d_{\max}$ is above the SNR threshold required to ensure the successful detection of the beacon at the MU. Moreover, we assume that the acknowledgment (ACK) from the MU is received perfectly by the BS (for instance, by using a conventional microwave technology as a control channel \cite{Rangan}).
Thus, for tractability, we assume that the misdetection probability is zero, and leave the more general analysis for future work.
It follows that we can express $P_{TX,k}$ as $P_{TX,k}=\rho_{TX}\omega_k$,
where $\rho_{TX}$ is the power per unit radiant required 
to achieve the target SNR.

If the MU is located within the transmitted beam area, it receives the beacon successfully and transmits an ACK packet to the BS, denoted as $c_k=1$ in slot $k$, received within the end of the slot. Otherwise, the BS declares a timeout (in this case, $c_k=0$).
Afterwards, the BS continues sensing in subsequent time slots until the end of the sensing period.
\begin{figure}[!t]
\vspace{-2mm}
\centering
\includegraphics[width=\columnwidth]{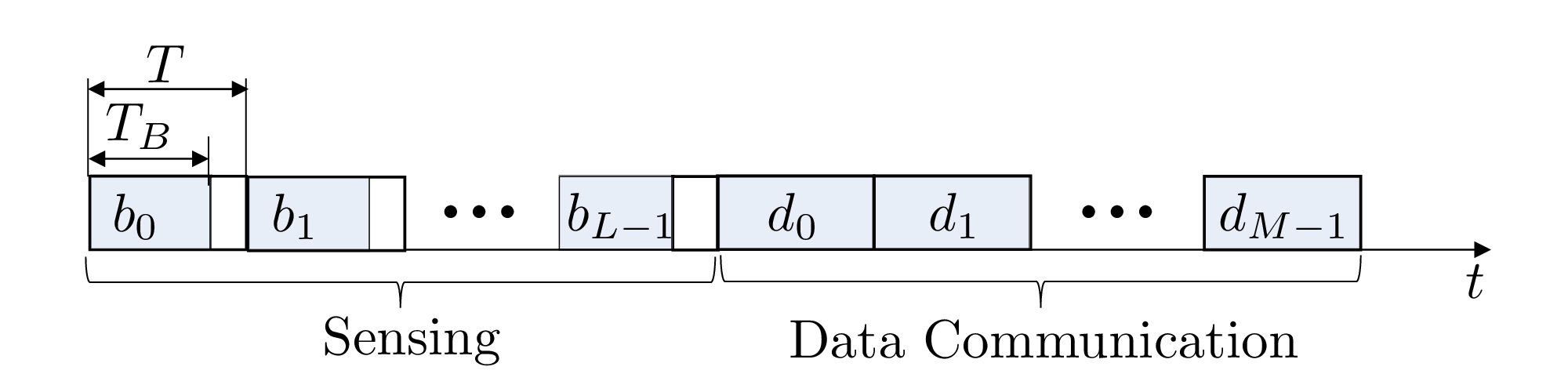}
\caption{The timing diagram of the sensing and data communication protocol.
$b_k$ denotes the sensing beacon in slot $k$, for $0\leq k<L$, whereas
$d_k$ denotes the data packet transmitted in slot $k+L$, for $0\leq k<M$.
}
\label{fig:timing_diagram}
\vspace{-7mm}
\end{figure}
\vspace{-4mm}
\section{Problem Formulation and Optimization}
In this section, we formulate the optimization problem as a Markov Decision Process (MDP)  \cite{bertsekas}, and optimize the sensing parameters to maximize the overall throughput over a sensing and data communication cycle. An MDP is defined by the 5-tuple $\{\mathcal T,\mathcal{S},\mathcal A,\mathbb{P}_k(S_k|S_{k-1},A_{k-1}),r_k(S_k,A_k),\forall k\in\mathcal T\}$, where $\mathcal T$ is the time horizon of the MDP, $\mathcal{S}$ is the state space, $\mathcal A$ is the set of actions, $\mathbb{P}_k(S_k|S_{k-1},A_{k-1})$ is the ensemble of transition probabilities given the previous state-action pair $(S_{k-1},A_{k-1})\in\mathcal S\times\mathcal A$, and $r_k(S_k,A_k)$ is the reward in slot $k$ given the state-action pair $(S_k,A_k)\in\mathcal S\times\mathcal A$.  
In our case, $\mathcal T{\equiv}\{0,1,\ldots,L\}$ represents the indexes of sensing time slots.
The slots $0\leq k<L$ correspond to the sensing phase, whereas, in slot $k=L$, the BS selects
 the beam parameters used in the data communication phase.
The state $S_k$ is 
the probability density function (PDF) of the angular coordinate $\Theta$ of the MU at the beginning of slot $k$, hence
\begin{align}
\int_{-\pi}^{\pi}S_k(\theta)\mathrm d\theta=1,\ S_k(\theta)\geq 0,\forall \theta\in[-\pi,\pi).
\end{align}
The action $A_k=[\alpha_k-\omega_k/2\;\;\alpha_k+\omega_k/2]$ specifies the beam pattern used in slot $k$ in the sensing phase (if $0\leq k<L$) and in the data communication phase (if $k=L$).
Thus, the action space is given by $\mathcal A\equiv \{[\alpha-\omega/2\;\;\alpha+\omega/2]:-\pi\leq\alpha<\pi,0<\omega \leq 2\pi\}$.
In slot $L$, the BS selects the beam parameters $A_L$ and transmits until slot $N$ using this beam.
We assume that the BS employs a fixed transmission power $P_{TX}$
in the data communication phase. This assumption presumes that 
this phase is the most energy demanding one of the entire transmission frame.
We define the reward function in slot $L$ as the throughput
achievable over one frame, per unit slot, \emph{i.e.},
\begin{align}
\label{eq:reward1}
&r_L(S_L,A_L){=}\frac{N{-}L}{N}\!\int_{A_L}\!\!S_L(\theta)\mathrm d\theta\log_2(1+\mathrm{SNR}(|A_L|)),
\end{align}
where $|A_k|=\int_{A_k} d\theta=\omega_k$ is the beam-width,
$(N-L)/N$ is the fraction of slots allocated to data communication,
and $\int_{A_L}S_L(\theta)\mathrm d\theta$
is the probability that the MU is inside the beam $A_L$. The $\mathrm{SNR}(|A_L|)$ is given as
\begin{align}
\label{SNReq}
\mathrm{SNR}(|A_L|) = \frac{\gamma_L}{|A_L|},\text{ where }\gamma_L{\triangleq}\frac{P_{TX}^{(L)}d_{\max}^{-\beta}}{2\pi N_0},
\end{align}
$P_{TX}^{(L)}$, $d_{\max}$, $\beta$, and $N_0$ denote the fixed transmission power of the BS over the data transmission slots, the maximum distance between the BS and the MU, the path loss exponent and the one-sided power spectral density of the noise component of the received signal, respectively. Herein, we assume that the noise is additive white Gaussian (AWGN). The term $2\pi$ in the denominator of
$\gamma_L$ in (\ref{SNReq})  corresponds to the omni-directional gain of the receiver, whereas $|A_L|$ in the denominator of 
$\mathrm{SNR}(|A_L|)$
corresponds to the directional gain at the BS, which is part of our design.
Moreover, we assume that the beacons duration $T_B\ll T$, thereby the total energy consumption in the sensing phase is small compared to that in the data communication phase. Thus, letting $P_{\mathrm{avg}}$ be the average power constraint over one frame, and assuming an equal transmission power allocation in the data communication phase, we obtain $P_{TX}^{(L)}=N\cdot P_{\mathrm{avg}}/(N{-}L)$.
For slots $k<L$ in the sensing phase, we have 
$r_k(S_k,A_k)=0$, since no throughput is accrued in these slots. However, these slots are functional to improving beam alignment in the data communication phase.
\vspace{-4mm}
\subsection{Transition Probabilities}
\label{sec:txprob}
At the beginning of the sensing phase, the belief is given by $S_0(\theta)=\frac{1}{\sigma}\chi(\theta\in[-\sigma/2,\sigma/2])$,
where $\chi(\cdot)$ is the indicator function,
since the angular coordinate is uniformly distributed over $[-\sigma/2,\sigma/2]$. We assume that $S_0$ is known.
Given the sequence of actions $A^{k-1}=(A_0,A_1,\dots,A_{k-1} )$ and of ACKs or timeouts $C^{k-1}=(C_0,C_1,\dots,C_{k-1})$, and the initial PDF (prior) $S_0$, the BS computes 
$S_k$ as 
\begin{align}
S_k(\theta)=f(\Theta=\theta|A^{k-1},C^{k-1}),
\end{align}
where $f(\cdot|\cdot)$ denotes the conditional PDF.
We let $U_k\triangleq\mathrm{supp}(S_k)$, where $\mathrm{supp}(f)$ denotes the support of $f$ over $[-\pi,\pi)$.
In particular, $U_0=[-\sigma/2,\sigma/2]$. Now, we can get
\begin{align}
\label{eq:postbelief}
&S_{k+1}(\theta)=f(\Theta=\theta|A^{k},C^{k-1},C_k=c_k)
\\&
\nonumber
\stackrel{(a)}{=}
\frac{\mathbb P(C_k=c_k|A^{k},C^{k-1},\Theta=\theta)
f(\Theta=\theta|A^{k},C^{k-1})
}
{\int_{-\pi}^\pi
\mathbb P(C_k=c_k|A^{k},C^{k-1},\Theta=\tilde\theta)
f(\Theta=\tilde\theta|A^{k},C^{k-1})\mathrm d\tilde\theta
}
\\&
\stackrel{(b)}{=}
\frac{
\mathbb P(C_k=c_k|A_k,\Theta=\theta)S_k(\theta)
}
{
\int_{-\pi}^\pi
\mathbb P(C_k=c_k|A_k,\Theta=\tilde\theta)S_k(\tilde\theta)
\mathrm d\tilde\theta
},
\end{align}
where in step (a) we have used Bayes'rule;
in step (b) we have used the fact that
$C_k=1\Leftrightarrow \theta\in A_k$, hence
$C_k=\chi(\theta\in A_k)$, which is thus a deterministic function of $A_k$ and $\theta$, independent 
of $(A^{k-1},C^{k-1})$; moreover, we have used the fact that 
$f(\Theta=\theta|A^{k},C^{k-1})=S_k(\theta)$ (since $\Theta$
is independent of $A_k$ given $(A^{k-1},C^{k-1})$).
Thus, $S_{k+1}$ is a function of $(S_k,A_k,C_k)$.
Similarly, 
\begin{align}
&\mathbb P(C_k=c_k|A^{k},C^{k-1})
\\&=
\nonumber
\int_{-\pi}^\pi\mathbb P(C_k=c_k|A^{k},C^{k-1},\Theta=\theta)
f(\Theta=\theta|A^{k},C^{k-1})\mathrm d\theta
\\&
=
\int_{-\pi}^\pi\mathbb P(C_k=c_k|A_k,\Theta=\theta)
S_k(\theta)\mathrm d\theta,
\end{align}
\emph{i.e.}, $C_k$ depends on
$(A^{k-1},C^{k-1})$ only through the current state-action pair $(A_k,S_k)$.
Thus,  we conclude that $S_{k+1}$ is statistically independent of $(S^{k-1},A^{k-1})$,
given $(S_k,A_k)$, and thus satisfies the Markov property.

In the following lemma, we provide a closed form expression of the belief $S_k$, and show that it can be expressed solely as a function of the initial belief $S_0$ and its support $U_k$.
\begin{lemma}
\label{lemma_2}
Given $S_0$ and $U_k=\mathrm{supp}(S_k)$,
the PDF of $\Theta$ in slot $k$, $S_k$,
is given by
\begin{align}
\label{eq:belief}
S_k(\theta)=\chi(\theta\in U_k)
\frac{S_0(\theta)}{\int_{U_k}S_0(\tilde\theta)\mathrm d\tilde\theta}.
\end{align}
\end{lemma}
\begin{proof}
We prove this lemma by induction. Clearly, we have
\begin{align}
S_0(\theta)=
\chi(\theta\in U_0)S_0(\theta)
=
\chi(\theta\in U_0)
\frac{S_0(\theta)}{\int_{U_0}S_0(\tilde\theta)\mathrm d\tilde\theta},
\end{align}
where we have used the fact that $\int_{U_0}S_0(\tilde\theta)\mathrm d\tilde\theta=1$ since $U_0=\mathrm{supp}(S_0)$.
Thus, $S_0$ can be expressed as (\ref{eq:belief}) with $k=0$.
Now assume $S_j$ is expressed as (\ref{eq:belief}) for some $j\geq 0$;
we show that $S_{j+1}$ is also expressed as (\ref{eq:belief}).
By letting 
\begin{align}
 A_j^{c_j} = 
\begin{cases}
 A_j, \text{ if }c_j = 1\\
 A_j^c,\text{ if }c_j = 0,
\end{cases}
\end{align}
and using the fact that $c_j=\chi(\theta\in A_j)$, from (\ref{eq:postbelief}) we have
\begin{align}
&S_{j+1}(\theta)=
\frac{
\chi(\theta\in A_j^{c_j})
S_j(\theta)}
{
\int_{-\pi}^\pi
\chi(\tilde\theta\in A_j^{c_j})S_j(\tilde\theta)
\mathrm d\tilde\theta
}
\end{align}
By using the induction hypothesis, we get
\begin{align}
S_{j+1}(\theta)=
\frac{
\chi(\theta\in A_j^{c_j}\cap U_j)
S_0(\theta)
}
{
\int_{-\pi}^\pi
\chi(\tilde\theta\in A_j^{c_j}\cap U_j)
S_0(\tilde\theta)
\mathrm d\tilde\theta
},
\label{xx}
\end{align}
Since $S_{j+1}(\theta){=}0$ outside of $A_j^{c_j}\cap U_j$,
we obtain $U_{j+1}\equiv A_j^{c_j}\cap U_j$ and
(\ref{eq:belief}) by substituting $U_{j+1}=A_j^{c_j}\cap U_j$
in (\ref{xx}).
\end{proof}
The implication of this lemma is that,
given the prior belief $S_0$ in slot $0$,
the support $U_k$ is a sufficient statistics. In fact, we can reconstruct the belief at time $k$ via (\ref{eq:belief}).
Importantly, this result holds even when $S_0$ is not an uniform distribution. Thus, in the following, we express the belief on the angular coordinate via the uncertainty set $U_k$. From the proof of Lemma \ref{lemma_2}, we note that the sequence $\{U_k,k\geq 0\}$ defining the support of $\{S_k,k\geq 0\}$ is obtained recursively as
\begin{align}
\label{recursion}
U_{k+1} = A_{k}^{c_{k}}\cap U_{k}.
\end{align}
\par Thus, when $\Theta \sim \text{Uniform}[-\sigma/2,\sigma/2]$,
with support $U_0=[-\sigma/2,\sigma/2]$,
from Lemma \ref{lemma_2}
we obtain
\begin{align}
\label{uniformcase}
S_k(\theta)=
\frac{\chi(\theta\in U_k)}{|U_k|}.
\end{align}
 We now investigate the form of the transition probabilities. If $C_k=1$, from (\ref{recursion}) we have that
$U_{k+1} = A_{k}\cap U_{k}$,
which occurs with probability
\begin{align}
&\mathbb P(C_k=1|U_k,A_k)=
\mathbb P(\Theta\in A_k|U_k,A_k)
\\&
=
\int_{A_k\cap U_k}S_k(\theta)\mathrm d\theta
=
\frac{|A_k\cap U_k|}{|U_k|},
\end{align}
where we have used the fact that $C_k=\chi(\Theta\in A_k)$,
and in the last step we used (\ref{uniformcase}).
On the other hand, if $C_k=0$, $U_{k+1} = A_{k}^{c}\cap U_{k}$,
which occurs with probability
\begin{align}
\mathbb P(C_k=0|U_k,A_k)=
1-\frac{|A_k\cap U_k|}{|U_k|}.
\end{align}
\vspace{-8mm}
\subsection{Optimization Problem and Value function Formulation}
We define the policy $\mu$ as $A_k=\mu_k(S_k)$
for $k=0,\dots,L$, which selects the beam parameters as a function of the PDF $S_k$ during the sensing and data communication phases.
The goal is to determine the optimal policy $\mu^*$ to maximize the throughput
$r_L(S_L,A_L)$, \emph{i.e.},
\begin{align}
\label{optprob}
\mu^*=\underset{\mu}{\textnormal{arg max}}
\ \mathbb E_\mu[r_L(S_L,A_L)|S_0].
\end{align}
Herein, we solve this optimization problem via dynamic programming \cite{bertsekas}. 
We denote the optimal  value function corresponding to the optimization problem (\ref{optprob}) as $V_k^*(U_k)$.

We derive the value function corresponding to each stage of the MDP and find the optimal beam parameters (\emph{i.e.}, the optimal policy $\mu^*$) for the sensing and data communication phases. Let $V_L(U_L,A_L)$ denote the value function at slot $k{=}L$
as a function of the state-action pair $(U_L,A_L)$.
Clearly, $ V_{L}(U_L,A_L){=}r_L(S_L,A_L)$,
with $S_L$ given by (\ref{uniformcase}),
and
$V_L^*(U_L){=}\underset{A_L\in\mathcal A}{\text{max}}V_L(U_L,A_L)$.
We 
obtain
\begin{align*}
\label{eq:val_fun_L}
V_{L}(U_L,&A_L)\leq V_{L}^*(U_L) =\max_{A_L\in\mathcal A} r_L(S_L,A_L)\\
&\stackrel{(a)}{=} \max_{A_L\in\mathcal A} 
\frac{N-L}{N}
\frac{|A_L\cap U_L|}{|U_L|} \log_2\left(1+\frac{\gamma_L}{|A_L|}\right)\\
&\stackrel{(b)}{\leq} \max_{0\leq|\tilde A_L|\leq |U_L|}
\frac{N-L}{N}
\frac{|\tilde A_L|}{|U_L|} \log_2\left(1+\frac{\gamma_L}{|\tilde A_L|}\right)\\
&\stackrel{(c)}{=}\frac{N-L}{N} \log_2\left(1+\frac{\gamma_L}{|U_{L}|}\right)\triangleq \tilde V_{L}^*(U_L),\numberthis
\end{align*}
where (a) follows by (\ref{eq:reward1}); (b) follows by letting $\tilde A_L \triangleq A_L \cap U_L \subseteq U_L$ and by using the fact that $|A_L|\geq |\tilde A_L|$; and (c) follows from the fact that the function 
\begin{align}
\tilde{V}_L(|U_L|,|\tilde A_L|) =  \frac{N-L}{N}\frac{|\tilde A_L|}{|U_L|} \log_2\left(1+
\frac{\gamma_L}{|\tilde A_L|}\right)
\end{align}
with $\gamma_L>0$
is a strictly increasing function of $|\tilde A_L|$ and hence,
it is maximized by $|\tilde A_L| = |U_L|$. The claim that $\tilde V_L(|U_L|,|\tilde A_L|)$ is strictly increasing in $|\tilde A_L|$ 
follows from the fact that
\begin{align*}
\frac{\partial \tilde V_L}{\partial |\tilde A_L|}=\frac{N-L}{N|U_L|\ln 2}\biggl[ \ln\left(1+\frac{\gamma_L}{|\tilde A_L| }\right) 
- \frac{\gamma_L/|\tilde A_L|}{1+\gamma_L/|\tilde A_L|} \biggr]>0
\end{align*}
since $\ln(1+y)> y/(1+y)$ for $y>0$.
Note that the upper bound $\tilde V_{L}^*(U_L)$ in (\ref{eq:val_fun_L})
can be achieved if $U_L$ is compact\footnote{Herein, since $U_k$ reflects angular coordinates, we define compactness up to a rotation of $2\pi$.} by choosing $A_L^*=U_L$,
which thus defines the optimal beam parameters $A_L^*$ for the
data communication phase when $U_L$ is compact.
This follows from the fact that, when $U_L$ is compact, then
$\tilde A_L\equiv A_L \cap U_L$ is a feasible beam, $\tilde A_L\in\mathcal A$.
More in general, for non compact $U_L$ we have the bound
\begin{align}
\label{upbound}
V_{L}^*(U_L)\leq
\tilde V_{L}^*(|U_L|) \triangleq \frac{N-L}{N} \log_2\left(1+\frac{\gamma_L}{|U_{L}|}\right),
\end{align}
and  thus, it is optimal to preserve compactness of $U_L$.
Hereafter, we will show that, indeed, the optimal sensing algorithm
preserves the compactness of $U_L$, so that the upper bound (\ref{upbound}) is attained.
We refer to $\tilde V_{k}^*(u)$ as the optimal value function under compactness constraint, \emph{i.e.}, achieved by a compact $U_k$ of size $|U_k|=u$. We have the following theorem.
\begin{theorem}
We have that
\begin{align*}
\label{eq:value_function_k}
V_k(U_k,A_k) \leq \tilde V_k^*(|U_k|)\ \forall U_k, A_k: 0\leq k \leq L \numberthis
\end{align*}
where
\begin{align}
\tilde V_k^*(|U_k|) = \frac{N-L}{N} \log_2\left(1+\frac{2^{L-k}\gamma_L}{|U_k|}\right).
\end{align}
The upper bound in (\ref{eq:value_function_k}) holds with equality if $U_k=[U_{k,\min},U_{k,\max}]$ is compact
and $A_j{=}A_j^*,\;\forall j{:}k{\leq}j{\leq}L$, where 
\begin{align}
\label{eq:beam1}
A_j^* = \left[U_{j,\min}\ ,\ \frac{U_{j,\min}+U_{j,\max}}{2}\right], k\leq j< L
\end{align}
or
\begin{align}
\label{eq:beam2}
A_j^* = 
\left[\frac{U_{j,\min}+U_{j,\max}}{2}\ ,\ U_{j,\max} \right], k\leq j< L,
\end{align}
where $U_{j,\min},U_{j,\max}$ are the extremes of the \emph{compact}
intervals $U_j=[U_{j,\min},U_{j,\max}],k\leq j\leq L$,
and 
\begin{align}
\label{eq:beam3}
A_L^* = U_L.
\end{align}
\end{theorem}
\begin{proof}
First, note that, if $U_k$ is compact
and $A_k$ is given by (\ref{eq:beam1})-(\ref{eq:beam3}),
then $U_{k+1}$ is compact. This directly follows by (\ref{recursion}). Thus, by induction, if 
$U_k$ is compact
and $A_j$ are given by (\ref{eq:beam1})-(\ref{eq:beam3})
for $j\geq k$, then $U_j$ are compact for $j\geq k$.

We prove the theorem by induction.
In (\ref{eq:val_fun_L}) and subsequent discussion, we have proved 
the claim of the theorem
and the fact that $A_L^* = U_L$ when $U_L$ is compact for the case $k=L$.
Thus, $V_L(U_L,A_L) \leq \tilde V_L^*(|U_L|)$,
and the upper bound is achievable when $U_L$ is compact and
$A_L=U_L$.
Now, let $k<L$ and assume that 
$V_{k+1}(U_{k+1},A_{k+1}) \leq \tilde V_{k+1}^*(|U_{k+1}|)$,
with upper bound achievable when $U_{k+1}$ is compact  and $A_j$ are given by (\ref{eq:beam1})-(\ref{eq:beam3}) for $j\geq k+1$. 
This hypothesis has been already proved for $k=L-1$.
We show that this implies 
$V_{k}(U_{k},A_{k}) \leq \tilde V_{k}^*(|U_{k}|)$,
achievable  when $U_k$ is compact and $A_j$ are given by (\ref{eq:beam1})-(\ref{eq:beam3}) for $k\leq j\leq L$.
The value function in slot $k$ as a function of the
state-action pair $(U_{k},A_{k})$ satisfies
\begin{align*}
\label{eq:val_fun_L_1_g}
&V_{k}(U_{k},A_{k})
\\&
=
\mathbb{E}[V_{k+1}^*(U_{k+1})|U_{k},A_{k}]
\leq \mathbb{E}[\tilde V_{k+1}^*(|U_{k+1}|)|U_{k},A_{k}]\\
&=\frac{N-L}{N}\biggl[
\frac{|A_{k}\cap U_{k}|}{|U_{k}|} \log_2\left(1+\frac{2^{L-1-k}\gamma_L}{|A_{k}\cap U_{k}|}\right)\\ 
&+ \left(1-\frac{|A_{k}\cap U_{k}|}{|U_{k}|}\right) \log_2\left(1+\frac{2^{L-1-k}\gamma_L}{ |A_{k}^c\cap U_{k}|}\right)\biggr],\numberthis
\end{align*}
where $V_{k+1}^*(U_{k+1})=\max_{A_{k+1}}V_{k+1}(U_{k+1},A_{k+1})$, and we have used the induction hypothesis.
In the last step, we have used the fact that $U_{k+1}=A_k\cap U_k$ with probability $|A_k\cap U_k|/|U_k|$, otherwise $U_{k+1}=A_k^c\cap U_k$ (see Sec. \ref{sec:txprob} and (\ref{recursion})).
From the induction hypothesis, equality holds above if $U_{k+1}$ is compact
and $A_j$ are chosen as in (\ref{eq:beam1})-(\ref{eq:beam3}) for $k+1\leq j\leq L$.
Letting $\tilde A_{k} \triangleq A_{k} \cap U_{k}\subseteq U_{k}$,
it then follows that
\begin{align*}
\label{eq:val_func_L_1}
V_{k}(U_{k},A_{k})&\leq \frac{N-L}{N}\biggl[\frac{|\tilde A_{k}|}{|U_{k}|} \log_2\left(1+\frac{2^{L-1-k}\gamma_L}{|\tilde A_{k}|}\right)\\ 
&+ \left(1-\frac{|\tilde A_{k}|}{|U_{k}|}\right) \log_2\left(1+\frac{2^{L-1-k}\gamma_L}{ |U_{k}|-|\tilde A_{k}|}\right)\biggr],\numberthis
\end{align*}
where we have used the fact that $\tilde A_{k}^c \cap U_{k} =  (A_{k}^c\cup U_{k}^c)\cap U_{k}=A_{k}^c\cap U_{k}$,
hence $|A_{k}^c\cap U_{k}|=|U_k|-|\tilde A_k|$. The function $\log_2(1+x)$ is concave in $x$, hence its perspective $t\log_2(1+x/t)$
is concave in $(x,t)$ for $x\geq 0$, $t>0$ \cite{boyd}.
Thus, by applying Jensen's inequality with $t_1\in(0,1)$ and $t_2=1-t_1$ we obtain
\begin{align*}
\frac{1}{2}t_1&\log_2\left(1+\frac{x}{t_1}\right)+
\frac{1}{2}t_2\log_2\left(1+\frac{x}{t_2}\right)
\\&
\leq
\frac{t_1+t_2}{2}\log_2\left(1+\frac{x}{\frac{t_1+t_2}{2}}\right)=\frac{1}{2}\log_2\left(1+2x\right).\numberthis
\end{align*}
By using this inequality with $x=\frac{2^{L-1-k}\gamma_L}{|U_{k}|}$,
$t_1=\frac{|\tilde A_{k}|}{|U_{k}|}$ and $t_2=1-t_1$, we 
can upper bound (\ref{eq:val_func_L_1}) as
 \begin{align*}
  V_{k}(U_{k},A_{k})\leq \frac{N-L}{N} \log_2\left(1+\frac{ 2^{L-k}\gamma_L}{|U_{k}|}\right)=\tilde V_{k}^*(|U_{k}|).
 \end{align*}
By inspection, this upper bound can be attained with equality if $U_{k}$ is compact and $A_{j} = A_{j}^*$ are given by (\ref{eq:beam1})-(\ref{eq:beam3}), $\forall j\geq k$.
 The induction step and the theorem are thus proved.
\end{proof}
\vspace{-5mm}
Since $U_0=[-\sigma/2,\sigma/2]$ is compact,
it can be inferred from Theorem 1 that the policy $A_j^*$
defined by (\ref{eq:beam1})-(\ref{eq:beam3})
is sufficient to preserve the compactness of subsequent $U_j,\ \forall j >0$ and is optimal.
By using Theorem 1, we can get $V_0^*(\sigma)$ as
\begin{align*}
&V_0^*(\sigma)=
\max_{\mu}\mathbb E_\mu[r_L(S_L,A_L)|S_0]{=}\frac{N-L}{N}\log_2 \left(1{+}\frac{2^{L}\gamma_L}{\sigma}\right)
\\&=\frac{N-L}{N}\log_2 \left(1{+}\frac{N2^{L}\gamma_0}{\sigma(N-L)}\right),\numberthis
\label{optimalperf}\end{align*}
where $\gamma_L$ is given by (\ref{SNReq}) and $\gamma_0 = (N-L)\gamma_L/N$. In the following, we express the dependence of $V_0^*(\sigma)$ on $(\sigma,L)$ as $V_0^*(\sigma,L)$.
Thus, $V_0^*(\sigma,L)$ denotes the maximum average throughput
achievable 
under the assumption that a portion $L$ of $N$ slots are allocated
for sensing.
\par Herein, we maximize $V_0^*(\sigma,L)$ with respect to $0\leq L\leq N$, for a given pair $(\sigma,N)$,
by solving the optimization problem
\begin{align}
L^* = \underset{L\in\{0,1,\ldots,N\}}{\textnormal{arg max}}\qquad 
V_0^*(\sigma,L).
\end{align}
The following theorem proves structural properties of 
$V_0^*(\sigma,L)$, which can be useful to optimize $L$ numerically.

\begin{lemma}
The function $V_0^*(\sigma,L)$ is a strictly log-concave 
function of $L\in [0,N]$.
\end{lemma}
\begin{proof}
Let $f(L) = \ln (V_0^*(\sigma,L))$, and $\zeta = \frac{(N-L)\gamma_L}{N\sigma}$, then
\begin{align*}
\numberthis
 &\frac{d^2 f(L)}{d L^2} = \frac{\frac{2^L\zeta}{N-L}\left(\frac{1}{N-L}+\ln 2\right)^2}{\left(1+\frac{2^L\zeta}{(N-L)}\right)\ln\left(1+\frac{2^L\zeta}{(N-L)}\right)}\\
 &\times \left[1-\frac{\frac{2^L\zeta}{N-L}}{1+\frac{2^L\zeta}{N-L}}\left(1+\frac{1}{\ln\left(1+\frac{2^L\zeta}{(N-L)}\right)}\right)\right]
 \end{align*}
\begin{align*}
 &+ \frac{2^L\zeta}{(N-L)^3}\frac{1}{\left(1+\frac{2^L\zeta}{(N-L)}\right)\ln\left(1+\frac{2^L\zeta}{(N-L)}\right)}-\frac{1}{(N-L)^2}\\
& \stackrel{(a)}{\leq} \frac{\frac{2^L\zeta}{(N-L)^3}}{\left(1+\frac{2^L\zeta}{(N-L)}\right)\ln\left(1+\frac{2^L\zeta}{(N-L)}\right)} -\frac{1}{(N-L)^2}
\stackrel{(b)}{<}0,
\end{align*}
where (a) follows from the fact that $\ln(1+x)\leq x$ for $x\geq 0$; (b) follows from $\ln(1+y)> y/(1+y)$ for $y>0$. 
Since $\frac{d^2 f(L)}{d L^2}<0$, then $f(L)$ is strictly concave in $L$, which implies that $V_0^*(\sigma,L)$ is strictly log-concave in $L$.
\end{proof}
Since $V_0^*(\sigma,L)$ is strictly log-concave,
we can first find $\hat L$ which maximizes $\ln (V_0^*(\sigma,L))$ over $[0,N]$ via convex optimization algorithms.
Then, the optimal $L^*$ over the \emph{discrete} set $\{0,1,\dots,N\}$ is obtained by solving
\begin{equation}
L^* = \underset{L\in\left\{\lceil\hat L \rceil,\lfloor\hat L\rfloor\right\}}{\textnormal{arg max}}\quad V_0^*(\sigma,L).
\end{equation}
\vspace{-3mm}
\section{Performance Analysis}
In this section, we compare the proposed algorithm (also referred to as \emph{bisection algorithm}), exhaustive search algorithm, and iterative search algorithm in terms of throughput performance. Herein, we assume that in the exhaustive search the BS scans at most $K$ consecutive non-overlapping sectors within $U_0$ each having width of $\sigma/K$. Once the MU is detected, the remaining slots are used for communication with the same beam pattern corresponding to the sensing slot when the MU was detected. In the iterative search, the BS divides $U_0$ into $M$ consecutive non-overlapping sectors each having size of $\sigma/M$, and scans at most $M-1$ regions to determine the sector where the MU is located. After finding it, the BS divides this sector into $M$ non-overlapping sectors each having width of $\sigma/M^2$ and scans at most $M-1$ sectors to locate the sector containing the MU. This process continues until the end of the sensing phase.

Note that the bisection algorithm is
a special case of the iterative one with $M=2$. Moreover, since the bisection algorithm has been optimized via dynamic programming, it always outperforms the iterative one.
On the other hand, the exhaustive algorithm has random sensing duration,
as opposed to the fixed sensing duration of the bisection algorithm,
since the BS scans different sectors until the MU is detected. Despite this inherent difference, in the next section we prove analytically that the bisection algorithm outperforms the exhaustive one as well, for all values of the sensing duration $L$.
\vspace{-5mm}
\subsection{Bisection versus Exhaustive Search}
Let's consider the exhaustive search algorithm where the MU receives isotropically and the BS uses $K\leq N$ non-overlapping beam patterns, each of width $\sigma/K$. Therefore, $K$ is the maximum duration of the sensing phase. The probability that the MU is detected in slot $J{=}j,0{\leq}j{\leq}K-1$ is $\mathbb P(J{=}j){=}1/K$, hence the average sensing duration is $\hat L\triangleq 1+\mathbb E[J]=(K+1)/2$. Therefore, 
the average throughput under exhaustive search is given by
\begin{align*}
\hat V_0(\sigma,\hat L) &{=} 
\mathbb{E}_J\left[\frac{N-J-1}{N} \log_2\left(1+\frac{NK\gamma_0}{(N-J-1)\sigma}\right)\right]\\
&\stackrel{(a)}{<}
\frac{N-\hat L}{N} \log_2\left(1+\frac{N(2\hat L - 1)\gamma_0}{(N-\hat L)\sigma} \right)\numberthis\\
&\stackrel{(b)}{\leq}
\frac{N-\lfloor\hat L\rfloor}{N} \log_2\left(1+\frac{N2\lfloor\hat L\rfloor\gamma_0}{(N-\lfloor\hat L\rfloor)\sigma} \right),
\numberthis\label{questo2}
\end{align*}
where in (a) we used the fact that $\log(1+x)$ is concave in $x$, hence its perspective function $t \log(1+x/t)$ is concave in $t>0$  \cite{boyd}, and Jensen's inequality; in (b), we used the fact that
$t \log(1+x/t)$
is an increasing function of $t$,
and 
$2\hat L-1\leq 2\lfloor\hat L\rfloor$ since $\hat L\in\{\lfloor\hat L\rfloor,\lfloor\hat L\rfloor+1/2\}$.
By comparing the throughput under the bisection and exhaustive search algorithms, we obtain the following lemma.
\vspace{-1mm}
\begin{lemma}
The bisection algorithm strictly outperforms the exhaustive one.
\end{lemma}
\vspace{-4mm}
\begin{proof}
Since $\hat L$ may not be an integer,
 we compare the performance of exhaustive search with that of the bisection algorithm with sensing duration $L=\lfloor\hat L\rfloor$.
 From (\ref{optimalperf}) and
 (\ref{questo2}),
 the performance gap between the two algorithms is given by
\begin{align*}
&V_0^*(\sigma,L)-\hat V_0(\sigma,\hat L)
\\&>
\frac{N{-}L}{N}
\left[\log_2\!\left(1{+}\frac{N2^{L}\gamma_0}{\sigma(N{-}L)}\right){-}
\log_2\!\left(1{+}\frac{N2L\gamma_0}{\sigma(N{-}L)}\right)\right]{\geq}0
\end{align*}
since $2^L\geq 2L$.
The lemma is thus proved.
\end{proof}
\begin{figure}[!t]
\centering
\vspace{-5mm}
\includegraphics[width=\columnwidth]{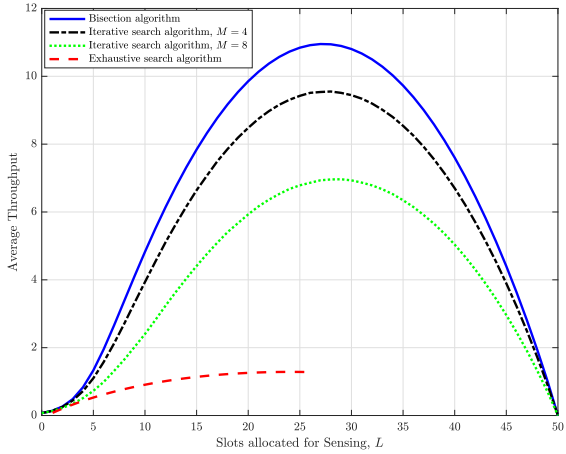}
\vspace{-6mm}
\caption{The average throughput versus sensing duration, $L$; $\gamma_0=-5\mathrm{dB}$, $\sigma = 2\pi$, $N=50$.
For exhaustive search, $1\leq K\leq N$.
}\vspace{-7mm}
\label{fig:throughput}
\end{figure}
\subsection{Numerical Results}
We consider the following scenario: $N=50$, $\gamma_0=-5\mathrm{dB}$, $\sigma=2\pi$.
In Fig.~\ref{fig:throughput}, we plot the throughput achieved by bisection, iterative, and exhaustive search algorithms
as a function of the sensing duration $L$.
Note that the throughput curves exhibit a quasi concave trend. It can also be noticed that the curve corresponding to the proposed bisection algorithm
achieves superior performance with respect to the exhaustive and iterative search algorithms, as proved analytically. 
Of particular interest is to compare the "peak" throughput of these algorithms, obtained by optimizing over the sensing duration $L$. We observe a performance degradation of approximately 12.8\% and 36.4\% for
the iterative algorithm with $M=4$ and $M=8$, respectively, compared to the bisection algorithm. Similarly, the peak throughput performance of the exhaustive algorithm is 88.3\% smaller than that of the bisection algorithm.
\section{Conclusion}
In this paper, we have studied the design of the optimal beam alignment algorithm in mm-wave downlink networks,
so as to maximize the throughput. We have proved
the optimality of a bisection algorithm,
and showed that it outperforms other algorithms proposed in the literature, such as exhaustive search and iterative search.
Moreover, we have formulated an optimization problem to find the optimal duration of the sensing phase in order to maximize the throughput, and we have shown that the iterative
algorithms with division factors of 4 and 8
and the exhaustive search algorithm 
achieve 12.8\%, 36.4\% and 88.3\% lower "peak" throughput 
than the bisection algorithm, respectively.
\bibliographystyle{IEEEtran}
\bibliography{IEEEabrv,biblio}

\end{document}